\documentclass[11pt]{article} %
\usepackage{pgfplots}
\pgfplotsset{width=7cm,compat=1.9}
\usepackage{authblk}
\usepackage{booktabs}
\usepackage{float}
\usepackage[letterpaper,margin=1in]{geometry}
\usepackage[scr=rsfso]{mathalfa}

\usepackage{setspace}
\usepackage{parskip} %

\usepackage{mdframed} %
\usepackage{microtype} %
\usepackage{mathpazo}
\usepackage{amsthm}
\usepackage{xspace}
\usepackage{comment}
\usepackage{lipsum}
\usepackage{amsfonts}
\usepackage{graphicx}
\usepackage{epstopdf}
\usepackage{algorithm}
\usepackage[noend]{algorithmic}
\usepackage{tikz}
\usepackage{caption}
\usepackage{subcaption}
\usepackage{amsmath}
\usepackage{hyperref}
\usepackage{amssymb}
\usepackage{dirtytalk}

\usepackage[colorinlistoftodos,textsize=tiny,textwidth=2cm,color=green!50!gray]{todonotes}
\usetikzlibrary{decorations.pathmorphing}
\ifpdf
  \DeclareGraphicsExtensions{.eps,.pdf,.png,.jpg}
\else
  \DeclareGraphicsExtensions{.eps}
\fi

\newcommand{\cfl}{\textsc{CFL}\xspace}
\newcommand{\ucfl}{\textsc{U-CFL}\xspace}
\newcommand{\ucflo}{\textsc{U-CFL-O}\xspace}
\newcommand{\cost}{\textrm{cost}\xspace}

\let\epsilon=\varepsilon %

\newtheorem{theorem}{Theorem} 
\newtheorem{lemma}[theorem]{Lemma} %

\usepackage[inline]{enumitem}
\setlist[enumerate,1]{label=(\roman*)}
\setlist[enumerate,2]{label=(\alph*)}

\newlist{stepsroman}{enumerate}{1}
	\setlist[stepsroman]{label=(\roman*)}

\newlist{stepsarabic}{enumerate}{1}
	\setlist[stepsarabic]{label=\arabic*., ref=\arabic*}

\newlist{stepsalph}{enumerate}{1}
	\setlist[stepsalph]{label=(\alph*), ref=(\alph*)}

\usepackage[capitalize, nameinlink]{cleveref}
\crefname{theorem}{Theorem}{Theorems}
\Crefname{lemma}{Lemma}{Lemmas}
\Crefname{claim}{Claim}{Claims}
\Crefname{fact}{Fact}{Facts}
\Crefname{remark}{Remark}{Remarks}
\Crefname{observation}{Observation}{Observations}
\Crefname{figure}{Figure}{Figures}
\Crefname{line}{Line}{Lines}
\Crefname{algocf}{Algorithm}{Algorithms}
\crefalias{AlgoLine}{line}
\Crefname{stepsromani}{Step}{Steps}
\Crefname{stepsarabici}{Step}{Steps}
\Crefname{stepsalph}{Step}{Steps}

\AddToHook{env/lemma/begin}{\crefalias{theorem}{lemma}}
\AddToHook{env/claim/begin}{\crefalias{theorem}{claim}}
\AddToHook{env/observation/begin}{\crefalias{theorem}{observation}}
\AddToHook{env/definition/begin}{\crefalias{theorem}{definition}}

\def\DEBUG{true} 
\ifdefined\DEBUG

\newcommand{\ramr}[1]{\todo[color=blue!100!white!50]{R: #1}}
\newcommand{\zacr}[1]{\todo[color=orange!100!black!50]{Z: #1}}
\newcommand{\fabr}[1]{\todo[color=red!100!black!50]{F: #1}}
\newcommand{\mrsr}[1]{\todo[color=green!100!black!50]{M: #1}}
\newcommand{\kinr}[1]{\todo[color=cyan!100!black!50]{K: #1}}
\else

\newcommand{\ramr}[1]{}

\newcommand{\fabr}[1]{}

\newcommand{\zacr}[1]{}

\newcommand{\mrsr}[1]{}

\newcommand{\kinr}[1]{}

\fi

\usepackage[
backend=biber,
style=numeric,
citestyle=numeric,
maxalphanames=6,
maxcitenames=99,
mincitenames=98,
maxbibnames=99,
giveninits=true,
url=false,
doi=true,
isbn=true,
backref=true
]{biblatex}

\begin{document}


\title{Multiple Swap Local Search Approximations for Uniform-Cost Capacitated Facility Location with Extensions to Outliers} 

\author{Zachary Friggstad}
\author{Pengfei Wu}

\affil{Department of Computing Science, University of Alberta}

\date{}

\maketitle

\thispagestyle{empty}
\begin{abstract}
We study \textsc{Capacitated Facility Location}, a fundamental clustering problem, in the setting of uniform opening costs and non-uniform capacities.
We give a polynomial-time $(3+\epsilon)$-approximation for this problem.
Our algorithm is a multiple-swap generalization of the local search algorithm by Dabas, Garg, and Gupta (2024), who obtained a $3.732$-approximation.
The main technical contribution is the analysis of this natural multiple-swap local-search method: variants of this approach have been analyzed previously, but allowing larger swaps requires new grouping and counting arguments.
We also show that a multiple-swap local search algorithm provides a $(3+\epsilon)$-approximation for \textsc{Capacitated Facility Location with Outliers}, again in the same setting of uniform opening costs and non-uniform capacities.
This is an improvement over the $6.372$-approximation, again by Dabas, Garg, and Gupta.
Finally, we show that the factor $3$ is tight for our local-search framework, and that the same multiple-swap idea does not improve the known locality gap for the non-uniform opening-cost setting.
\end{abstract}

\newpage
\setcounter{page}{1}

\section{Introduction}
Facility location problems are among the most extensively studied models in operations research and clustering. At a high level, the goal is to place service providers and connect clients to these service providers while minimizing costs. In this paper, we focus on the case where each facility can only serve a limited number of clients.
Specifically, in the \textsc{Capacitated Facility Location Problem} (\cfl), we are given a finite collection of $n$ points $V$ that is partitioned into subsets $V_C \subseteq V$ (the clients) and $V_F \subseteq V$ (the facilities) along with metric distances $d(i,j), i,j \in V$.
Each $i \in V_F$ has an associated opening cost $f_i \geq 0$ and capacity $u_i \geq 0$.

The goal is to select some facilities $S \subseteq V_F$ and an assignment $\sigma : V_C \rightarrow S$ satisfying $|\sigma^{-1}(i)| \leq u_i$ for each $i \in S$ while minimizing $\sum_{i \in S} f_i + \sum_{j \in V_C} d(j, \sigma(j))$.
That is, we want to minimize the cost of opening all facilities plus the cost of connecting clients to their prescribed facilities. We emphasize that we are only allowed to open each facility once, i.e. we consider the {\em hard-capacitated} setting.

For a given $S \subseteq V_F$, finding the cheapest assignment $\sigma : V_C \rightarrow S$ can be done in polynomial time by solving a minimum-cost transportation problem: throughout this paper we let $\sigma_S$ denote any one particular minimum-cost assignment for $S$.
For brevity, for any $S \subseteq F$ we define
\[ \cost(S) := \left\{
\begin{array}{cl}
\displaystyle \sum_{i \in S} f_i + \sum_{j \in V_C} d(j,\sigma_S(j)) & \text{ if } \sum_{i \in S} u_i \geq |V_C| \\
\infty & \text{ otherwise}
\end{array}\
\right.
\]
So, compactly, the goal is to find $S \subseteq V_F$ minimizing $\cost(S)$.
Following a series of work on this problem (see Section \ref{sec:prev}), Bansal, Garg, and Gupta \cite{bansal} gave the current-best approximation for \cfl with an approximation guarantee of $5$.

In this paper, we focus on the special case of \cfl where all facilities have the same opening cost. That is, in \textsc{Capacitated Facility Location Problem with Uniform Opening Costs} (\ucfl) there is a single value $f$ such that $f_i = f$ for each $i \in V_F$.
Prior to our work on this problem, the best approximation algorithm was a $3.732$-approximation by Dabas, Garg, and Gupta \cite{CFLIPCO} using a local search algorithm. We still consider non-uniform capacities, it is just the opening costs that are identical.

Finally, we also consider on the case where not all clients need to be connected in a solution. In the \textsc{Capacitated Facility Location with Outliers} problem, we are additionally given a value $z \geq 0$ indicating the number of {\em outliers} permitted by a solution.
In a solution, we only need to connect $|V_C|-z$ clients to open facilities, the remaining clients are discarded at no cost. Note that given $S \subseteq V_F$, we can find cheapest assignment that assigns at least  $|V_C|-z$ clients to $S$ in a capacity-respecting fashion in polynomial time, again by solving a minimum-cost transportation problem where we include a dummy node that can accept up to $z$ units of demand with no cost. Let $Z_S \subseteq V_C$ be the clients that are discarded and $\sigma_S : V_C-Z_S \rightarrow S$ the assignment satisfying $|\sigma^{-1}(i)| \leq u_i$ in the best solution for $S$.

Again, we let $\cost(S)$ denote the cost of the cheapest assignment that discards at most $z$ clients, it will be clear from the context when we are referring to the version of \cfl with outliers and when we are not.
Let \ucflo denote the variant of \cfl with uniform opening costs and outliers. The best approximation for \ucflo prior to our work is a $6.372$-approximation based on local search \cite{CFLIPCO}.


\subsection{Our Results and Techniques}\label{sec:results}
Our main results are the following.

\begin{theorem}\label{thm:ucfl}
    There is a $(3+\epsilon)$-approximation for \ucfl for any constant $\epsilon > 0$.
\end{theorem}
\begin{theorem}\label{thm:ucflo}
    There is a $(3+\epsilon)$-approximation for \ucflo for any constant $\epsilon > 0$.
\end{theorem}
Of course, Theorem \ref{thm:ucflo} subsumes Theorem \ref{thm:ucfl}. We first prove Theorem \ref{thm:ucfl} as both a warm up to the ideas and because we follow the framework set out in \cite{ZacOutlier} for adapting local-search based approximations for uniform opening cost facility location problems to the extension with outliers.

As in \cite{CFLIPCO} and many previous algorithms for \cfl, our algorithms are based on local search.
At a high level, our local search algorithm begins with $S = V_F$ and iteratively considers one of two operations: either adding a single facility and recomputing the optimal assignment, or adding a set $S'\subseteq V_F$ to $S$, where $S'$ contains at most a constant number $\rho$ of facilities, and closing an arbitrary number of facilities $T$ from $S$ while limiting the assignment-cost recalculation to optimally reassign clients in $\sigma_S^{-1}(T)$ to $S'$. Specific details are described in Section~\ref{sec:Pre}. Note that our algorithm is a natural extension of the algorithm in \cite{CFLIPCO}: it attempts to open a large constant number of facilities in the swap operation, instead of just one or two as in their \ucfl and \ucflo algorithms, respectively. This larger-swap operation is natural to consider, but its analysis requires new grouping and counting arguments. In contrast to previous multi-exchange local-search approaches such as \cite{bansal, multiexchange}, where the local moves are organized around opening or closing a single facility and then using an auxiliary optimization problem to choose the corresponding set of facilities, our swap operation directly allows a large set of facilities to be opened in a single move while identifying a suitable set of currently open facilities to close. We believe that this analysis, together with our improvement and extension of the outlier framework of~\cite{ZacOutlier} to the capacitated setting, may be useful for studying related facility-location variants.

We show a locally-optimal solution has cost at most $3+o_\rho(1)$ times the optimum solution's cost. Theorem \ref{thm:ucfl} is obtained by selecting $\rho$ large enough as a function of $1/\epsilon$ and using a standard trick to ensure the algorithm only uses a polynomial number of local search steps. In Section~\ref{sec:tight}, we also show that this factor is tight for our local-search approach, and that the analogous multiple-swap method does not improve the known locality gap for non-uniform opening costs.

As with most local search algorithms for \cfl, our analysis considers an exchange graph, namely the one described in \cite{CFLIPCO}. In their argument, they bound the opening costs of facilities in one of two ways. The opening costs of facilities in $S$ that are leaves in the exchange graph are bounded using swap operations that only open one new facility $i$ and close all leaves connected to $i$. The opening costs of facilities in $S$ that are internal nodes in the exchange graph are bounded using a simple counting argument (recalling that opening costs are uniform) which does not require any test swaps. Our improvement comes from grouping many swaps together to bound the opening cost of leaves and also ``low-degree' internal nodes. We also rely on a counting argument to ensure the unaccounted facilities in $S$ (having high degree in the exchange graph) to bound the number of high-degree facilities in $S$ to $O(\epsilon) \cdot (|S| + |S^*|)$ where $S^*$ is some fixed optimal solution. The precise details are a bit more involved since it also requires careful handling of low- versus high-degree nodes of the optimum solution as well in the exchange graph and a second, unrelated counting argument is used even in the analysis of the test swaps.



Finally, to prove Theorem \ref{thm:ucflo} we use the framework in \cite{ZacOutlier} to group together swaps used in our algorithm for \ucfl to prove, with a larger but still constant value $\rho'$, that the same set of local search operations (appropriately generalized to the outliers case) yield a $(3+\epsilon)$-approximation for \ucflo. Note that there is no black-box result that can be cited from \cite{ZacOutlier} as they merely provide a template for adapting local search to handle outliers and show its effectiveness on key examples. We are demonstrating that an appropriate adaptation works in our setting, though the adaptation is not immediate. We view the resulting grouping technique and its extension to the outlier setting as potentially useful building blocks for future local-search algorithms for related facility-location variants.


\subsection{Previous Work}\label{sec:prev}
The case of uniform opening costs we consider here was first explicitly studied by Levi, Shmoys, and Swamy who gave a 5-approximation based on rounding an LP relaxation \cite{levi}. Then Aardal, van den Berg, and Gijswijt gave a $4.562$-approximation based on local search \cite{aardal}.
Kao revisited LP rounding to give a $4$-approximation \cite{kao}. Finally, Dabas, Garg, and Gupta demonstrated that local search can be used to get a $3.732$-approximation \cite{CFLIPCO} which is the best approximation prior to our work.

For general \cfl, the first constant-factor approximation for \cfl was given by P\'{a}l, Tardos, and Wexler \cite{pal} and was subsequently improved by Mahdian and Pal \cite{mahdian} (for the more general \textsc{Universal Facility Location} problem) and then by Zhang, Chen, and Yi \cite{multiexchange}.
Later, Bansal, Garg, and Gupta gave the current-best algorithm with an approximation ratio of $5$ \cite{bansal}. While the natural LP relaxation for \cfl with non-uniform opening costs has an unbounded integrality gap, approximation algorithms based on stronger LP relaxations were first given by An, Singh, and Svensson \cite{an} and later improved to a 9.0927 LP-based approximation by Kao \cite{kao}.

The only non-bicriteria previous work on \cfl with outliers is the aforementioned $6.372$-approximation in the case of \ucflo by Dabas, Garg, and Gupta \cite{CFLIPCO}, while Dabas and Gupta provided a bicriteria algorithm for \cfl with outliers\slash penalties~\cite{DabasG25}. As mentioned earlier, Friggstad et al. give a template that helps convert local search algorithms for uniform-cost facility location problems to the variant with outliers \cite{ZacOutlier}, though they only show it works in the facility location problems without capacities.


\section{The Local Search Algorithm for \ucfl and the Exchange Graph}\label{sec:Pre}
In this section and the next, we focus on \ucfl. The details for \ucflo are presented in Section \ref{sec:outliers}. Throughout, we assume $\sum_{i \in V_F} u_i$ is at least the number of clients that are required to be served, i.e. at least $|V_C|$ for \ucfl and at least $|V_C|-z$ for \ucflo, otherwise there can be no feasible solution.

Our local search algorithm maintains some $S \subseteq V_F$, initially setting $S = V_F$ itself. In each iteration, it tries one of the following operations:
\begin{itemize}
\item {\bf add($i)$}: If there is some $i \in V_F - S$ such that $\cost(S \cup \{i\}) < \cost(S)$, we update $S \gets S \cup \{i\}$.
\item {\bf swap($S',T$)}: Consider some $S' \subseteq V_F$ with $|S'| \leq \rho$ and some $T \subseteq S$. Suppose we first close $T$ (i.e. remove from $S$) and then open $S'$, noting any $i \in S' \cap S$ remains open. Then we compute the optimal way to assign all clients $j \in \sigma_S^{-1}(T)$ from their current location $\sigma_S(j)$ to $S'$ while ensuring each $i \in S'$ receives at most $u_i$ clients.
If this results in a solution cheaper than $\cost(S)$, we update $S \gets (S - T) \cup S'$.
\end{itemize}
If there are multiple choices for such an operation, we may choose any. We emphasize that after some application of {\bf swap($S',T$)}, we still recompute $\sigma_S$ for the new $S \subseteq V_F$ in preparation for future iterations, but determining if {\bf swap($S',T$)} can be performed
depends only on the reassignment of $\sigma^{-1}(T)$ to $S'$.

Concerning the running time, we first explain how a single iteration of the local search algorithm can be executed to run in polynomial time. We can simply enumerate all $i \in V_F - S$ and compute $\cost(S \cup \{i\})$ in polynomial time to determine if such an operation would find a cheaper solution.
Searching for an admissable {\bf swap($S',T$)} operation is more involved. There are $O(n^\rho)$ subsets $S'$ to be tried, we enumerate them all. Then we solve an auxiliary optimization problem to find the best $T$, i.e. the set $T \subseteq S$ that would lead to the most improved solution after
reassigning $\sigma^{-1}(T)$ optimally to $S$, accounting also for the change $f \cdot |S'| - f \cdot |T|$ in opening cost. This can be viewed as a variant of the multiple knapsack problem where every item (i.e. center that we may consider closing) has polynomially-bounded integer demands that can be integrally split between multiple knapsacks. A standard dynamic programming algorithm is used to find the best $T$ in $n^{O(\rho)}$ time. The full details appear in Appendix \ref{app:dp}.

While each iteration runs polynomial time, it is not clear that a local optimum will be reached after polynomially-many iterations. The usual trick works here. Namely, our analysis uses at most $2n$ ``test swaps'' so if we limit ourselves to operations that only find solutions with cost $\leq (1 - \epsilon/(2n))\cdot \cost(S)$ then we only lose an $\epsilon$-factor in the approximation guarantee. The number of iterations is now polynomial since every $2n/\epsilon$ iterations will result in the cost dropping by a constant factor. These details are entirely standard in local search analysis and are omitted.

From now on, let $S$ be a local optimum solution and $S^*$ be any fixed optimal solution to the \ucfl instance, recall that $\sigma_S$ and $\sigma_{S^*}$ denote optimal client assignments in the respective solutions. Throughout we let $F = f \cdot |S|$ be the opening cost and $C = \sum_{j \in V_C} d(j, \sigma_S(j))$ be the assignment cost of solution $S$. Similarly, let $F^* = f \cdot |S^*|$ and $C^* = \sum_{j \in V_C} d(j, \sigma_{S^*}(j))$

~

\noindent
{\bf The Exchange Graph}\\
As in many previous works on \cfl, we use an {\em exchange graph} to analyze the local search procedure. In our case, we use the exchange graph introduced in \cite{CFLIPCO}, who build off of the exchange graph used in \cite{pal}.
That is, in \cite{CFLIPCO} a bipartite graph $H = (S \dot\cup S^*, E)$ is described where each $i \in S \cap S^*$ has a copy on each side of $H$. We adopt the convention that when we refer to an edge $(i,i^*) \in E$ that $i \in S$ and $i^* \in S^*$.
Additionally, integer values $x(i,i^*) \geq 0$ for $(i,i^*) \in E$ are also given. The graph $H$ and the vector $x$ have the following properties
\begin{enumerate}
\item $H$ is a forest.
\item $\sum_{(i,i^*) \in E} x(i,i^*) \cdot d(i,i^*) \leq C + C^*$.
\item $x(\delta(i)) = |\sigma_S^{-1}(i)|$ for $i \in S$ and $x(\delta(i^*)) = |\sigma_{S^*}^{-1}(i^*)|$ for $i^* \in S^*$.
\item For each $i \in S \cap S^*$, $(i,i) \in E$ and one of these copies of $i$ in $H$ is a leaf. 
\end{enumerate}
Here, $\delta(\cdot)$ denotes the standard graph-theoretic cut set. At a high level, $H$ and $x$ are constructed by solving a minimum-cost transportation problem sending $|\sigma_S^{-1}(i)|$ clients from each $i \in S$ to $S^*$ so each $i^* \in S^*$ receives $|\sigma_{S^*}^{-1}(i^*)|$ clients. Edges $E$ are the support of an optimal solution, which can be taken to be acyclic by cycle-shifting arguments. The assignments $\sigma_S$ and $\sigma_{S^*}$ can be used to describe a transportation solution with cost at most $C + C^*$, so the optimal transportation solution is no worse. Finally, that one copy of each $i \in S \cap S^*$ is a leaf follows because we can assume an optimal solution sends $\min\{|\sigma^{-1}_S(i)|, |\sigma^{-1}_{S^*}(i)|\}$ clients from $i \in S$ to its copy $i \in S^*$ in $H$, thereby making one of them a leaf node.



\section{\ucfl Analysis}\label{section:CFL}
As in \cite{CFLIPCO} and most other local search work in \cfl, the following holds simply because no ${\bf add}$ operation results in a better solution.
\begin{lemma}[Lemma 4.1, Arya et al.~\cite{Arya}]
$C \leq C^* + F^*$
\end{lemma}
We remark that this analysis requires at most $n$ ``test swaps'', one per $i \in S^* - S$. In \cite{CFLIPCO}, the bound $F \leq 3F^* + 2C^*$ is given leading to a bound of $\cost(S) \leq 4 \cdot \cost(S^*)$. A facility cost rescaling trick is used to get their final approximation guarantee of $3.732$.

Here is where we deviate in our analysis. For some constant $\epsilon > 0$ such that $1/\epsilon$ is an integer, we now set $\rho := 2(1/\epsilon)^{2/\epsilon}$. We prove
\begin{theorem}\label{thm:fac}
$(1-2\epsilon) \cdot F \leq (2+2\epsilon) \cdot F^* + 2\cdot C^*$
\end{theorem}
From this, we get $\cost(S) = C + F \leq \frac{3\cdot C^* + (3+2\epsilon) \cdot F^*}{1-2\epsilon} \leq (3+9\epsilon) \cdot \cost(S^*)$ (for sufficiently small $\epsilon$).
This completes the proof of Theorem \ref{thm:ucfl} after shrinking $\epsilon$ by a factor of 9.

Theorem \ref{thm:fac} is proven by identifying a collection of at most $n$ different {\bf swap$(S',T)$} operations based on the exchange graph. Consider a component of $H$. There must be some facility in $S^*$ in this component, otherwise it is a singleton component corresponding to some $i \in S-S^*$ with $\sigma^{-1}_S(i) = \emptyset$ so {\bf swap$(\emptyset, \{i\})$} would be an improving operation as the cost would decrease by exactly $f$ for closing $i$.
We view each component of $H$ as being rooted at some facility in $S^*$.

Let $L \subseteq S$ be the facilities in $S$ that are leaf nodes in $H$. For each $i \in S \dot\cup S^*$, let $\mathcal C(i)$ be the children of $i$. 
Let $S_s = \{i \in S : |\mathcal C(i)| \leq 1/\epsilon\}$ and $S_\ell = \{i \in S : |\mathcal C(i)| > 1/\epsilon\}$, this partitions $S$ into facilities with a ``small'' number of children and a ``large'' number of children.
In a similar but not identical manner, we define $S^*_s = \{i \in S^* : |\mathcal C(i) - L| \leq 1/\epsilon\}$ and $S^*_\ell = S^* - S^*_s$ partition $S^*$ into facilities with a ``small'' number of {\em non-leaf} children and a ``large'' number of non-leaf children.

\begin{lemma}\label{lem:big}
$f \cdot |S_\ell| \leq \epsilon \cdot f \cdot |S^*|$
\end{lemma}
\begin{proof}
This is simply because every $i \in S_\ell$ has at least $1/\epsilon$ children in $S^*$. This does not double-count, i.e. $\mathcal C(i) \cap \mathcal C(i')=\emptyset$ for any distinct $i,i' \in S \dot\cup S^*$.
\end{proof}

All that is left is to bound the opening cost of $S_s$. To do this, we consider the subgraph $H'$ of $H$ obtained by deleting $S_\ell \cup L$ (only the copy from $S$ but leaving the copy in $S^*$ if it was in both), i.e. the induced subgraph $H' := H[(S - S_\ell - L) \dot\cup S^*]$. Each component of $H'$
can still be viewed as being rooted in $S^*$ since the original trees of $H$ were and we only deleted nodes in $S$. We will decompose $H'$ into subtrees with limited overlap, each containing at most $\rho$ nodes from $S^*$.

Let $R_0, R_1, R_2, \ldots$ denote the layers of nodes of $S^*$ in the trees of $H'$. That is, $R_0$ is all root nodes in $H'$ and, inductively, for $m \geq 1$ we have $R_m$ is all grandchildren of nodes in $R_{m-1}$. For some $\Delta \in [0,1/\epsilon)$, consider $B_\Delta := R_0 \bigcup \cup_{m' \geq 0} R_{\Delta + m' / \epsilon}$. Since the $R_m$ partition $S^*$, there is some $\Delta$ such that $|B_\Delta| \leq |R_0| + \epsilon \cdot |S^*|$, fix any such $\Delta$.

Now, for each $i^* \in B_\Delta \cup S^*_\ell$ consider the maximal subtree of $H'$ rooted at $i^*$ that does not contain any internal nodes from $B_\Delta \cup S^*_\ell$. That is, consider all nodes encountered in a depth-first search in the subtree under $i^*$ that does not recurse deeper from another node in $B_\Delta \cup S^*_\ell$. Let $T'_{i^*}$ denote this subtree. Note each $T'_{i^*}$ has depth at most $2/\epsilon$ and the only node in $T'_{i^*}$ that can have more than $1/\epsilon$ children in $T'_{i^*}$ is $i^*$ itself, but only if $i^* \in S^*_\ell$.

Next, for each $i^* \in B_\Delta \cup S^*_\ell$ we create subtrees of $T'_{i^*}$ as follows. Consider the children of $i^*$ in $T'_{i^*}$, i.e. $\mathcal C(i^*) - S_\ell - L$. Arbitrarily partition these children so that each part has size exactly $1/\epsilon$ except, perhaps, the last part which has size at most $1/\epsilon$.
For each part $\pi$ of the partition, we consider the subtree of $T'_{i^*}$ rooted at $i^*$ but only including the subtrees rooted at children in that part $\pi$.

Let $\mathcal T'_{i^*}$ be the collection of subtrees of $T'_{i^*}$ created this way, one per part of the partition.
Note $|\mathcal T'_{i^*}| > 1$ can only happen if $i^* \in S^*_\ell$. Each tree in $\mathcal T'_{i^*}$ has height at most $2/\epsilon$ and each node has at most $1/\epsilon$ children, so the number of nodes is at most $2(1/\epsilon)^{2/\epsilon}$.
Finally let $\mathcal T' = \cup_{i^* \in B_\Delta \cup S^*_\ell} \mathcal T'_{i^*}$ denote all such trees created this way.
Observe each $i^* \in S^*$ appears in at least one tree in $\mathcal T'$ and can only appear in multiple if $i^* \in B_\Delta \cup S^*_\ell$. This construction is depicted in Figure~\ref{fig:trees}.

Finally, for each $i \in L$ pick a single tree in $\mathcal T'$ that contains the parent of $i$ and add $i$ along with its only incident edge to this tree. Let $\mathcal T$ denote these final trees, each tree in $\mathcal T$ is a tree in $\mathcal T'$ with, perhaps, some leaves in $L$ appended to it. By construction, each $i \in S_s$ appears in precisely one tree of $\mathcal T$: if $i \in L$ it was only included in the final formation of $\mathcal T$ from $\mathcal T'$ and if $i \in S_s-L$ it was an internal node for precisely one tree in $\mathcal T'$.
Finally, we note each edge $(i,i^*)$ of the original exchange graph $H$ itself is contained in at most one subtree, in fact in exactly one subtree if and only if $i \in S_s$.

\begin{figure}
    \centering
    \includegraphics[width=0.5\textwidth]{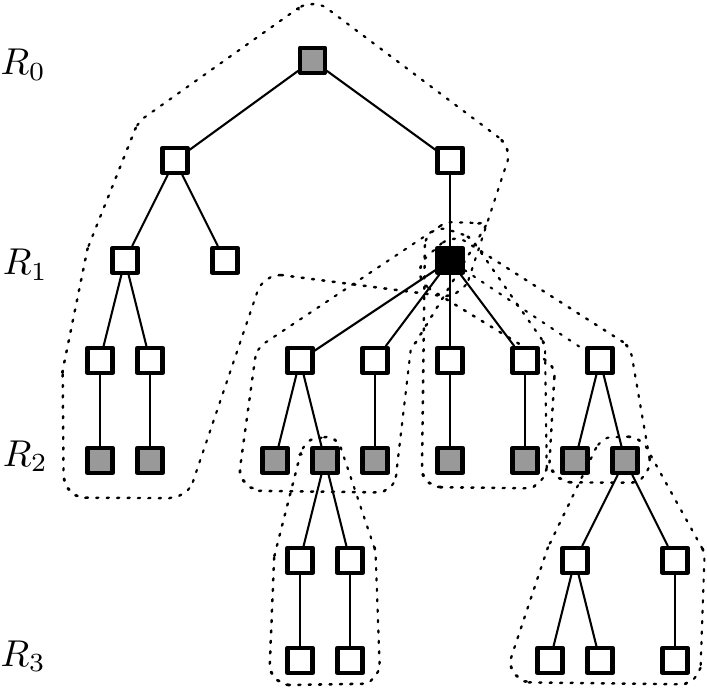} 
    \caption{A decomposition of $H'$ into subtrees $\mathcal T'$ (the leaves $L$ added to form $\mathcal T$ are not depicted) for $1/\epsilon = 2$. Nodes in layer labeled by some $R_i$ in $S^*$, nodes between these layers are in $S_s$. The grey nodes are in $B_\Delta$ (using $\Delta = 0$) and the black node is in $S^*_\ell$.
    The subtrees $\mathcal T'$ are indicated by dashed contours, each node in $B_\Delta-S^*_\ell$ is the root of exactly one subtree, the only node in $S^*_\ell$ is the root of three subtrees. If we used $\Delta = 1$ instead, then we would still have the top node in $B_\Delta$ by construction.}
    \label{fig:trees} 
\end{figure}

Consider a tree $\tau \in \mathcal T$ and, for brevity, let $S' = \tau \cap S^*$ and $T = \tau \cap S$. As noted above, $\tau$ has at most $\rho = 2(1/\epsilon)^{2/\epsilon}$ nodes not in $L$, so $|S'| \leq \rho$ meaning the cost change calculated in {\bf swap$(S',T)$} is nonnegative as $S$ is a local minimum.

The optimal cost of reassigning $\sigma^{-1}(T)$ to $S'$ in {\bf swap($S',T$)} is upper bounded by the following reassignment.
Note that each $i \in T$ has each of its neighbours in $S'$ by construction of the subtrees. That is, if $(i,i^*) \in E$ then $i^* \in S'$. Hence it is feasible to move all of $\sigma^{-1}(i)$ to $i^*$ for each $i \in T$ by property (3) of $H$ from Section \ref{sec:Pre}.
The cost of this reassignment is at most $x(i,i^*) \cdot d(i,i^*)$. Summing over all $i \in T$, the reassignment cost is at most the
total $x$-weighted distance of the edges of $H$ contained in this subtree.

Recalling no such swap will result in an improved solution by local optimality of $S$, summing over all subtrees and using property (2) of $H$ from Section \ref{sec:Pre} shows:
\begin{equation}\label{eqn:bound}
0 \leq f \cdot |S^*| + \sum_{i^* \in B_\Delta \cup S^*_\ell} f \cdot |\mathcal T'_{i^*}| - f \cdot |R_0| - f \cdot |S_s| + C + C^*.
\end{equation}
Here, each $i^* \in S^* - R_0$ has been counted once in the first term when it is a non-root of the corresponding subtree in $\mathcal T'$ and is counted as many times in the sum as there are subtrees rooted at $i^*$ in $\mathcal T'$. Each $i^* \in R_0$ is only opened $|\mathcal T'_i|$ times, once for each subtree for which it is the root, because it is not in any tree in $\mathcal T$ rooted at a shallower node. Its contribution to $|S^*|$ is offset by its negative contribution to $|R_0|$.

There is a subtlety in how the opening costs are considered for some $i \in S_s \cap S^*$. By property (4) of $H$, we have $(i,i) \in E$ and one copy of $i$ in $H$ is a leaf. If $i \in S$ is a leaf, then the swap that closes it also opens its parent $i$: the $-f$ appears in front of $|S_s|$ for this swap
and the $+f$ is accounted for in the term $|S^*|$. But there may be other swaps involving $i \in S^*$, i.e. if $i^*$ is in $B_\Delta$ or $S^*_\ell$. Since $i$ is already open when $i \in S'$ for these swaps, the upper bound on cost calculated by {\bf swap}($S',T$) for these swaps
does not include the $+f$ for this $i \in S'$. However, it is safe to include its contribution in the sum because this is an upper bound. If $i \in S^*$ is a leaf then the only actual swap is when $i \in S$ is closed as $\mathcal T'_{i} = \emptyset$ and, again, the $-f$ for closing $i$ is offset by the $+f$ for opening $i$
in the upper bound.

The final step is to bound the sum using a counting argument.
\begin{lemma}\label{lem:tech}
$\sum_{i^* \in B_\Delta \cup S_\ell^*} |\mathcal T'_{i^*}| - |R_0| \leq \epsilon \cdot |S^*| + 2\epsilon \cdot |S|$
\end{lemma}
\begin{proof}
If $i^* \in B_\Delta - S^*_\ell$ then $|\mathcal T'_{i^*}| \leq 1$; it could be zero if there are no non-leaf children of $i^*$ in $H'$. So $\sum_{i^* \in B_\Delta - S^*_\ell} |\mathcal T'_{i^*}| \leq |B_\Delta| \leq |R_0| + \epsilon \cdot |S^*|$.

For $i^* \in S^*_\ell$, $|\mathcal T'_{i^*}| = \lceil |\epsilon \cdot \mathcal C(i^*)| \rceil \leq 2\epsilon \cdot |\mathcal C(i^*)|$.
Since each $i \in S$ is the child of exactly one facility in $S^*$, $\sum_{i \in S^*_\ell} |\mathcal T'_{i^*}| \leq 2\epsilon \cdot |S|$ which completes the proof.
\end{proof}
Applying this to \eqref{eqn:bound} and rearranging shows
\[
f\cdot |S_s| - 2\epsilon \cdot f \cdot |S|
\leq (1+\epsilon) \cdot F^* + C + C^*.
\]
Applying Lemma~\ref{lem:big} to the preceding inequality gives
\[
(1-2\epsilon) \cdot f \cdot |S|
\leq (1+2\epsilon) \cdot F^* + C + C^*
\leq (2+2\epsilon) \cdot F^* + 2 \cdot C^*,
\]
which completes the proof of Theorem~\ref{thm:fac}.


\section{\ucflo Algorithm and Analysis}\label{sec:outliers}
Recall that for each $S \subseteq V_F$, we can compute $\cost(S)$ along with an optimal way to discard at most $z$ clients and assign the remaining clients to $S$ in polynomial time. We let $|Z_S|$ denote the clients that are discarded and $\sigma_S : V_C - Z_S \rightarrow S$ be the assignment of the rest. Of course, we assume $z < |V_C|$ as otherwise we may discard all clients. Given this, we may then assume $|Z_S| = z$ as well since we could always discard clients until exactly $z$ are discarded without increasing the cost of a solution.

The local search operations we consider are the following. Let $\rho'$ be some constant that we will fix later.
\begin{itemize}
\item \textbf{add$(i)$} for $i \in V_F - S$: This operation is the same as in Section~\ref{section:CFL}, with the understanding that our computation of $\cost(S)$ is now for the outliers case.
\item \textbf{swap$(S',T)$}: For $S' \subseteq V_F$ with $|S'| \leq \rho'$ and $T \subseteq S$, this operation first closes $T$ and then opens $S'$ as before (leaving any $i \in S' \cap T$ open). It then finds an optimal way to assign some clients in $\sigma_S^{-1}(T)$ to $S'$ while discarding the rest as outliers, and to connect some clients in $Z_S$ to $S'$, while ensuring that the number of outliers remains at most $z$. The clients $j \in \sigma_S^{-1}(T)$ that are reassigned to $S'$ pay the additional distance from $\sigma_S(j)$ to their new facility in $S'$ when we compute the cost of this reassignment. We show how to compute the optimum swap operation of this form in polynomial time in Appendix \ref{app:dp}.
\end{itemize}

As before, let $S$ now denote a local optimum solution, $S^*$ a global optimum solution, and $C,F,C^*,F^*$ the respective client/facility costs of these solutions. We emphasize $Z_S$ and $Z_{S^*}$ are the outliers in the respective solutions, that $|Z_S| = |Z_{S^*}| = z$, and that $\sigma_S, \sigma_{S^*}$ are the assignments of the non-outliers in the respective solutions. We assume $Z_S \cap Z_{S^*} = \emptyset$, as $S$ would still be a local optimum in the instance with $Z_S \cap Z_{S^*}$ removed and $z$ adjusted to $z - |Z_S \cap Z_{S^*}|$, and $S^*$ would still be a global optimum solution in this adjusted instance. Finally, for each $i \in S$, let $Z_{S^*}(i) = \sigma_S(i) \cap Z_{S^*}$, i.e. the clients assigned to $i$ in the local optimum solution that are outliers in the global optimum. Similarly for $i^* \in S^*$, let $Z_S(i^*) = \sigma_{S^*}(i^*) \cap Z_S$.

The bound on the assignment cost of a locally optimal solution obtained from the \textbf{add} operation follows from previous work, and is very similar to the classic argument in \cfl.
\begin{lemma} [Lemma 6 in~\cite{CFLIPCO}]\label{Lemma:OutCcost}
    $C \leq C^* + F^*$
\end{lemma}

Dabas, Garg, and Gupta introduce a slightly different exchange graph for their outlier analysis; we use that exchange graph. Here, $H = (S \dot\cup S^*, E)$ is a bipartite graph, and we have integer values $x(i,i^*) \geq 0$ for $(i,i') \in E$. We state this slightly differently, but it follows directly from their construction
(say, by deleting the dummy nodes introduced for outliers and simply handling the outliers explicitly in property (2) below). The properties of $H$ in this case are

\begin{enumerate}
\item $H$ is a forest.
\item $\sum_{(i,i^*) \in E} x(i,i^*) \cdot d(i,i^*) + \sum_{j \in Z_S} d(j, \sigma_{S^*}(j)) + \sum_{j \in Z_{S^*}} d(j,\sigma_S(j)) \leq C + C^*$.
\item $x(\delta(i)) = |\sigma_S^{-1}(i)| - |Z_{S^*}(i)|$ for $i \in S$ and $x(\delta(i^*)) = |\sigma_{S^*}^{-1}(i^*)|\cdot |- Z_S(i^*)|$ for $i^* \in S^*$.
\item For each $i \in S \cap S^*$, $(i,i) \in E$ and one of these copies of $i$ in $H$ is a leaf. 
\end{enumerate}

The total cost of the exchange graph is at most $C+C^*$.
Again, this follows immediately from their construction. Our phrasing is just a different view that is convenient for our algorithm.

\subsection{Bounding the Facility Cost}
If we were to form test swaps in the same way as in our analysis of \ucfl but also apply the outlier changes associated with each swap, then the main reason we cannot use them to generate a similar bound on $f \cdot |S|$ is that the number of outliers used in the analysis of the swaps may exceed $z$.
To fix this, we use the idea from Friggstad et al.~\cite{ZacOutlier} to combine the test swaps into larger groups each of which uses $\leq z$ outliers. There may be some overlap in these groups, but it will be limited enough to get an effective bound on $f \cdot |S|$.

Let $\rho = (1/\epsilon)^{2/\epsilon}$ be the same constant as in the previous section; we will eventually set $\rho'$ to be a constant larger than $\rho$.
Almost everything from the previous section still works here. We define $S_s, S_\ell, S^*_s, S^*_\ell$ as before and we still get the bound $f \cdot |S_\ell| \leq \epsilon\cdot |S^*|$, i.e. Lemma \ref{lem:big} works without adjustment in our setting.
We form $H'$ by deleting the nodes $i \in S_\ell$ (keeping their copies in $S^*$ if they are there), select $\Delta$ and form $B_\Delta$, and decompose $H'$ into subtrees $\mathcal T$ as before using this same value of $\epsilon$. Lemma \ref{lem:tech} holds without modification.

We now consider outliers. For each tree $\tau \in \mathcal T$ we create two sets $Z_S(\tau) \subseteq Z_S$ and $Z_{S^*}(\tau) \subseteq Z_{S^*}$ as follows.
For each $j \in Z_S$, pick any single tree $\tau \in \mathcal T$ with $\sigma_{S^*}(j) \in \tau$ and add $j$ to $Z_{S}(\tau)$. For each $j \in Z_{S^*}$ with $\sigma_{S}(j) \in S_s$, pick the unique tree $\tau \in \mathcal T$ with $\sigma_{S}(j) \in \tau$ and add $j$ to $Z_{S^*}(\tau)$. Note the latter step only considers clients that are assigned to a small facility in $S$. Therefore, $\{Z_S(\tau) : \tau \in \mathcal T\}$ partitions $Z_S$ and $\{Z_{S^*}(\tau) : \tau \in \mathcal T\}$ is a partition of a subset of $Z_{S^*}$. Clients $j \in Z_{S^*}$ not in a part (i.e. assigned to some $S_\ell$ by $\sigma_S$) are effectively ignored in the rest of this proof. The idea is that when we swap facilities in $\tau$, $Z_S(\tau)$ and $Z_{S^*}(\tau)$ will be the clients that change their outlier status in this swap.

We first consider what happens if we try to perform the swaps based on $\tau \in \mathcal T$ and their associated outlier sets. That is, if we were to close all $\tau \cap S$ and open all $\tau \cap S^*$ as before, we could upper bound the reassignment cost as follows. We could first make $Z_{S^*}(\tau)$
outliers and each $j \in Z_{S}(\tau)$ could be assigned to $\sigma_{S^*}(j)$ instead. The remaining clients located at centres in $T$ could be reassigned as per the $x$-values on the edges as in the previous section. If this left at most $z$ outliers, then this reassignment would place a valid upper bound on {\bf swap}($\tau \cap S^*, \tau \cap S$), and the cost analysis would be the same as in the previous section, except for the cost change for outliers, which is handled by property (2) of the exchange graph $H$ in this section.

To fix the issue that this way of analyzing swaps might not be ``balanced'' in terms of the outlier difference, we follow a strategy similar to that of \cite{ZacOutlier} by merging swaps, though there are some differences.

\subsection{Grouping Subtrees}

First, we partition $\mathcal T$ into groups as follows. Initially, let each $\tau \in \mathcal T$ be in a group of its own. Then, while there are two groups having at most $\rho/2$ facilities from $S^*$, merge them. Afterwards, if there is a group that has fewer than $\rho/2$ facilities from $S^*$, merge it with any group (there must be another group, otherwise $|S^*| \leq \rho$ so local search would have trivially found the optimal solution). Each group now has between $\rho/2$ and $3\rho/2$ facilities.

Let $\mathcal P$ be the resulting partition. Each $P \in \mathcal P$ is a collection of trees from $\mathcal T$. We use the notation $P \cap S$ to denote the $i \in S$ that appear among the trees in $P$ and $P \cap S^*$ to denote the $i \in S^*$ that appear at least once among the trees in $P$.
Let $Z_S(P) = \cup_{\tau \in P} Z_S(\tau)$ and $Z_{S^*}(P) = \cup_{\tau \in P} Z_{S^*}(\tau)$. Finally, let $\Delta_P = |Z^*_S(P)| - |Z_S(P)|$. The idea is that if we were to perform {\bf swap($P \cap S^*, P \cap S$)} and bound the cost using the tree edges plus connecting $Z_S(P)$ and letting $Z_{S^*}(P)$ be outliers, the number of outliers would be $z + \Delta_P$. Thus, if $\Delta_P \leq 0$, this would lead to a valid reassignment.
We also let
\[ E(P) = \{(i,i^*) \in E : (i,i^*) \text{ is an edge of some } \tau \in P\}.\]
Note no $(i,i^*) \in E$ is in more than one set of the form $E(P)$ since it appeared in at most one subtree $\tau \in \mathcal T$.

Let $\mathcal P^+ = \{P \in \mathcal P, \Delta_P > 0\}$, $\mathcal P^- = \{P \in \mathcal P, \Delta_P < 0\}$ and $\mathcal P^0 = \{P \in \mathcal P, \Delta_P = 0\}$.
We first show how to derive a bound on the reassignment cost by using swaps in $\mathcal P^0$.
\begin{lemma}\label{lem:zero}
If $\rho' \geq 3/2 \rho$, then for each $P \in \mathcal P^0$ we have
\[ f \cdot |P \cap S| \leq f \cdot |P \cap S^*| + \sum_{(i,i^*) \in E(P)} x(i,i^*) \cdot d(i,i^*) +  \sum_{j \in Z_{S^*}(P)} d(j,\sigma_S) \]
\end{lemma}
\begin{proof}
Note that $|P \cap S^*| \leq 3/2 \rho$, so we can consider {\bf swap}($P \cap S^*, P \cap S$). This opens each $i \in P \cap S^*$ and closes each $i \in P \cap S$. If some $i$ is in both sets, we may regard it as both closed and then opened, so it contributes to both sides of the bound to be proven.

To bound the optimal reassignment cost for this swap, we move clients $j \in Z_S(P)$ to their open facility in $P \cap S^*$ and have clients $j \in Z_{S^*}(P)$ now be outliers (we save their connection cost in $S$, but that is not needed for the bound).
The remaining clients in $\sigma_S(P \cap S) - Z_{S^*}(P)$ are reassigned according to $x$, so exactly $x(i,i^*)$ clients are moved from $i \in T$ to $i^* \in S'$ for each $(i,i^*) \in E(P)$ and this accounts for all client movement.
This is a feasible reallocation because all edges incident to each $i \in P \cap S$ are in $E(P)$, meaning all corresponding $i^*$ are opened. Together with property (3) of $H$, this guarantees there is enough room in each $i^* \in P \cap S^*$ to accommodate all incoming clients (the ones being reassigned from $P \cap S$ and the ones in $Z_S$ being assigned there). Since this places an upper bound on the optimal reassignment and since the swap operation did not result in a cheaper solution, the bound follows.
\end{proof}

If $\mathcal P = \mathcal P^0$, i.e. all parts were balanced, we would be done because summing this bound over all $P \in \mathcal P$ would make the left-hand side sum to exactly $f \cdot |S_s|$, the first term on the right would be $(1+\epsilon) f \cdot |S^*| + 2\epsilon \cdot f \cdot |S_s|$ as before, and the sum over $(i,i^*) \in E(P)$ would be at most the $x$-weighted cost of $H$, which, along with the total outlier bound, is at most $C+C^*$. That is, we would get the exact same bound as in the non-outliers setting.
In the next section, we discuss how to handle the case when there are some imbalanced parts.

\subsection{Managing the Imbalanced Parts}
Finally, we fix $\rho' := 3\rho/\epsilon$. Suppose, for now, that $\sum_{P \in \mathcal P} \Delta_P = 0$, i.e. all $j \in Z_{S^*}$ have $\sigma_S(j) \in S_s$ and hence are accounted for in $\mathcal P$. In \cite{ZacOutlier}, a bipartite graph $Q$ with sides $\mathcal P^+, \mathcal P^-$ and a bijection $\kappa : Z_{S^*} \rightarrow  Z_S$ is constructed. We use the following properties of this construction:
\begin{itemize}
\item The graph $Q$ is a caterpillar, i.e. a single ``central'' path plus some leaf nodes appended to the path.
\item For each $j \in Z_{S^*}$, the parts $P$ and $P'$ with $j \in Z_{S^*}(P)$ and $\kappa(j) \in Z_S(P')$ are either equal or are connected by an edge in $Q$.
\item The edges of $Q$ admit a natural linear ordering: the edges of the central path are ordered according to an arbitrary direction of the path, and all leaves incident to an internal part $P$ are ordered arbitrarily between the two path edges incident to $P$.
\item For any group $G_i$, let $P(G_i)$ be the parts incident to edges of $G_i$. The boundary parts
$B_i \subseteq P(G_i)$, namely those incident to edges outside $G_i$, satisfy $|B_i|\le 2$.

\item The local facilities in $B_i$ decompose into consecutive blocks in the left-to-right ordering. By the construction of $\kappa$, each block is incident to a single group $G_j$, and different blocks are internally disjoint, overlapping only possibly at boundary local facilities. Let $B'_{G_i}$ be the set of facilities in the blocks incident to $G_i$, excluding the two boundary facilities.
\end{itemize}
See Figure~\ref{fig:caterpillar} for an illustration of $Q$ and some concepts defined below.

\begin{figure}[h] 
    \centering
    \includegraphics[width=0.8\textwidth]{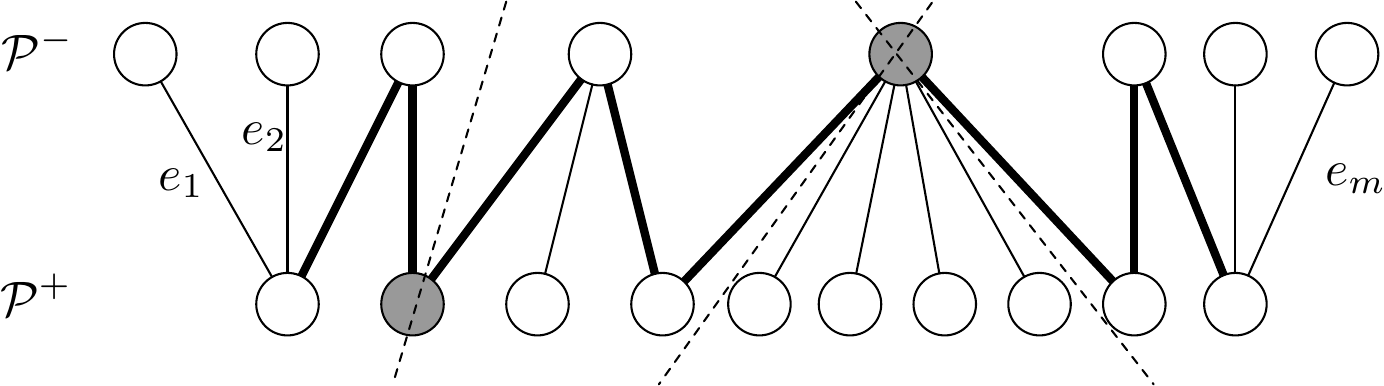} 
    \caption{The caterpillar graph $Q$. The central path is depicted with bold, the ordering of the edges is depicted ``left to right''. Groups $G_i$ formed using $1/\epsilon = 4$ are delimited by dashed lines, the shaded nodes correspond to parts that are in at least one boundary set $B_i$.}
    \label{fig:caterpillar} 
\end{figure}

In our case, we only have that $\sum_P \Delta_P \leq 0$. Still, the exact same construction works except $\kappa$ is only guaranteed to be one-to-one but not necessarily onto. This is a straightforward verification, the details are skipped.

Label the edges of $Q$ according to this ordering as $e_1, e_2, \ldots, e_m$. Finally, partition these edges into groups of $1/\epsilon$ consecutive edges except the last group, which should have between $1/\epsilon$ and $2/\epsilon$ edges (if the number of edges $m$ was less than $1/\epsilon$, then $|S^*| \leq 3\rho /(2\epsilon)$ so the $\rho'$-swap algorithm would have found an optimal solution). Call these groups $G_1, G_2, \ldots, G_{m'}$.

For a group $G_i$, let $E(G_i)$ be all edges of $H$ appearing in the parts of $P(G_i)$, and consider
\[
{\bf swap}\bigl(P(G_i)\cap S^*,\, (P(G_i)-B_i+B'_{G_i})\cap S\bigr).
\]
This operation opens every facility of $S^*$ that appears in a tree contained in a part of $P(G_i)$, and closes every facility of $S$ in $P(G_i)$ except those lying in a boundary part but not in $B'_{G_i}$.
Note that $|P(G_i) \cap S^*| \leq \rho'$ since there are at most $1/\epsilon + 1 \leq 2/\epsilon$ parts in each $P(G_i)$ and each has at most $3\rho/2$ facilities in $S^*$. To bound the cost change, consider the following reassignment:
\begin{itemize}
\item For every $j \in \cup_{P \in P(G_i)} Z_{S^*}(P)$ with $\kappa(j) \in  \cup_{P \in P(G_i)} Z_S(P)$ as well, make $j$ an outlier and instead connect $\kappa(j)$ to $\sigma_{S^*}(\kappa(j))$, which is now open.
\item All other $j$ with $\sigma_S(j) \in \cup_{P \in P(G_i)} P \cap S$ are redistributed according to $x$ to open facilities.
\end{itemize}
Recall that, for every facility that is closed, all of its neighbours in $H$ are opened because the groups are, ultimately, unions of trees $\tau \in \mathcal T$. Thus the redistribution is valid. Also note that, for every $i$ that is closed, each $j \in Z_{S^*}(i)$ has $\sigma_{S^*}(\kappa(j))$ open because we only close facilities if every neighbour (in $Q$) of its part is also involved in the swap. Thus, for every outlier we create we also reclaim one; therefore, the number of outliers after this reassignment has not increased. The cost of the described reassignment places an upper bound on the test swaps' cost. Let $B'_B$ denote the set of all boundary facilities, that is $B'_B=\left(\bigcup_iB_i\right)\setminus\left(\bigcup_{G_i}B'_{G_i}\right)$. We immediately get the following bound:
\begin{equation}\label{eqn:final}
 f \cdot \left|\left(P(G_i) \cap S\right)\setminus B_B'\right|\leq f \cdot |P(G_i) \cap S^*| + \sum_{(i,i^*) \in E(G_i)} x(i,i^*) d(i,i^*) + \sum_{j \in Z_{S^*}(P(G_i))} d(\kappa(j),\sigma_{S^*}(\kappa(j)))
 \end{equation}
where we have, again, ignored the savings from no longer assigning $j \in Z_{S^*}(P(G_i))$.

Now the only facilities that are never closed are the facilities in $B'_B$. The maximum number of such intervals over all groups is at most $\epsilon|S^*|$, and thus the maximum number of facilities that are never closed is at most $2\epsilon|S^*|$. Then $|B'_B|\leq 2\epsilon\cdot |S^*|$

Summing over all groups, we almost get the same bound as in the case of \ucfl. The only difference is that some parts $P$ are opened multiple times, i.e. whenever they are part of some $B_i$. But in every group, at least $1/\epsilon - 2 \geq 1/(2\epsilon)$ parts are not boundary parts. Since each part has size between $\rho/2$ and $3\rho/2$, in each group $i$ the number of facilities of $S^*$ in one of the at most two parts in $B_i$ is at most $3\rho$ and the number of facilities of $S^*$ (counted with multiplicity if they lie in multiple parts in the group) is at least $\rho/(4\epsilon)$.

Therefore, the total number of facilities we swap in over all boundary parts $B_i$ is at most
\[ 12\epsilon \cdot \sum_{\tau \in \mathcal T} |\tau \cap S^*| \leq 12\epsilon \cdot \left((1+\epsilon) \cdot |S^*| + 2\epsilon |S_s|\right) = O(\epsilon) (|S^*| + |S|)\]
where the last inequality follows from the analysis of \ucfl.

To conclude, if we sum \eqref{eqn:final} over all groups and apply Lemma \ref{lem:zero} to $\mathcal P^0$, we get
\[ f \cdot |S_s\setminus B'_B| \leq (1+O(\epsilon)) f \cdot |S^*| - O(\epsilon) \cdot |S_s| + C + C^* \leq (2+O(\epsilon)) f \cdot |S^*| - O(\epsilon) \cdot |S_s| + 2C^* \]
recalling property (2) of $H$ when aggregating all of the charges to $x$ and outlier connections.

Rearranging and adding the bounds $|B'_B|\leq 2\epsilon\cdot |S^*|$, $f \cdot |S_\ell| \leq \epsilon \cdot |S^*|$, and $C \leq C^* + F^*$ shows
\[ C+F \leq (3 + O(\epsilon)) (F^* + C^*) \]
as required.

\section{Tight Examples for Local Search}\label{sec:tight}
We have a simple example showing that our local search analysis is tight; see Figure~\ref{fig:tightexampleMulti}. The circles and rectangle are facilities, and the triangles are clients. All facilities have opening cost $2$, and the two circle facilities form a local optimum. Adding a facility does not decrease the total cost. Moreover, a swap operation also does not decrease the total cost: if we open the rectangle and close the two circle facilities, we save $2$ in opening cost, but rerouting the two clients from the circle facilities to the rectangle has total cost $2$. Thus, there is no decrease in the total cost.

The local optimum has cost $6$: opening the two circle facilities costs $4$, and the routing cost is $2$. The optimal solution has cost $2$: opening the rectangle costs $2$, and the routing cost is $0$. Thus, the ratio is $3$.

One might argue that the swap operation could save $2$ in routing cost if the clients were reassigned directly from their locations. However, the algorithm in \ref{app:dp} that computes the optimal set $T$ to close for a given subset of facilities $S'$ to open measures the assignment cost change in moving a client $j$ from, say, $i \in T$ to $i' \in S'$ as $d(i,i')$ rather than $d(j,i') - d(j,i)$. That is why our {\bf swap} operation for any constant $\rho$ would not find an improving swap in even in this simple example.

\begin{figure}[h] 
    \centering
    \includegraphics[width=1\textwidth]{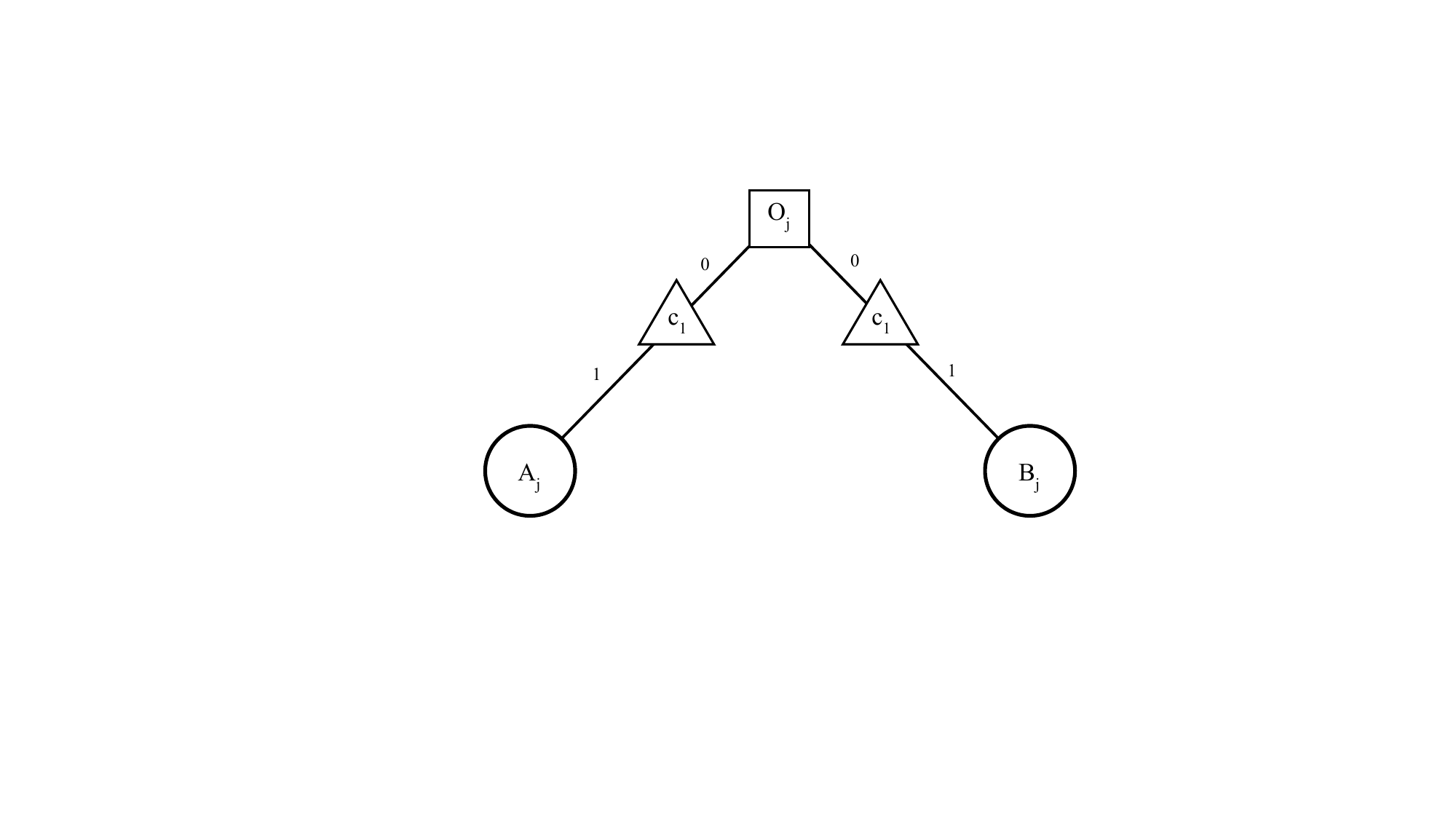} 
    \caption{The tight example.}
    \label{fig:tightexampleMulti} 
\end{figure}

\subsection{A Lower Bound for Non-Uniform Opening Costs}

One can easily check that we can still efficiently search for an improving swap move even if the facility opening costs are non-uniform. So one might wonder if our approach yields improvements for this setting.
Unfortunately the bad examples from previous work for local search for instances of CFL with non-uniform opening costs easily extend to give the same lower bounds for our multiple-swap algorithm. That is, our algorithm does not improve over the 5-approximation by Bansal, Garg, and Gupta \cite{bansal} for CFL with non-uniform opening costs.

Figure~\ref{fig:tightexample} is a modification of the bad locality gap examples from \cite{DBLP:journals/mor/ZhangCY05,bansal}. As in these papers, the rectangular boxes are unopened facilities, the circle boxes are the facilities opened in the local optimum solution, and the triangle boxes are client demands. For the rectangle boxes and circle boxes, the value above is the opening cost, and the value below is the capacity.

It is easy to see that opening a rectangle box does not decrease the solution cost. Consider any swap move that opens at most $\rho + 1$ facilities; we can assume that $\rho \ll n$. First, closing one circle facility and opening the rectangle below the circle together with any other non-root rectangle does not decrease the total cost: closing one circle saves $4$, but rerouting from the circle to another non-root rectangle has cost $4$. Therefore, this move does not decrease the total cost. Finally, if we open the root rectangle, then we would also open $\rho$ bottom rectangles and close all circles in the same subtree. This saves $4\rho$, but opening the root rectangle costs $4\rho$, and there is no saving in the connection cost. Thus, the circle facilities are indeed a local optimum solution.

Now, if we calculate the cost of the local solution, we have opening cost $8n$ for the circle facilities and connection cost $2n$. For the optimal solution, the opening cost for the rectangle facilities is $4\rho$, and the connection cost is $2n$. Thus, the ratio is
\[
\frac{10n}{4\rho + 2n}.
\]
As $n \to \infty$ and $\rho \ll n$, this becomes $5$.

\begin{figure}[h] 
    \centering
    \includegraphics[width=1\textwidth]{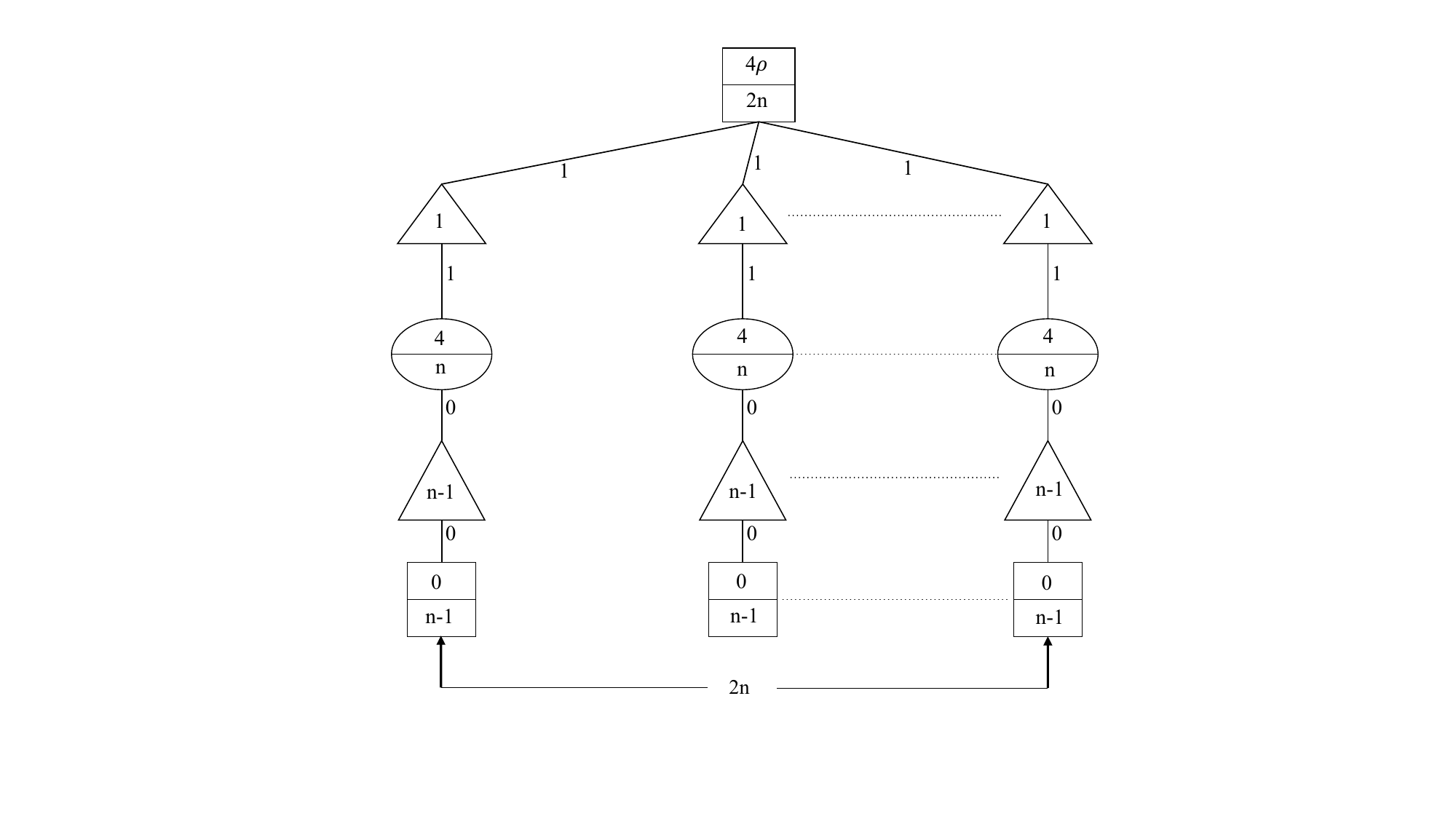} 
    \caption{The tight example.}
    \label{fig:tightexample} 
\end{figure}

\printbibliography


\appendix
\renewcommand{\thesection}{Appendix \Alph{section}}

\section{Finding the best {\bf swap($S',T$)} Operation} \label{app:dp}


\begin{lemma} \label{lemma:CFLSwap}
Given $S' \subseteq V_F$ with $|S'| \leq \rho$, we can find some $T \subseteq S$ minimizing the cost change of $\textbf{swap}(S',T)$ in the \ucfl algorithm in $n^{O(\rho)}$ time.
\end{lemma}
\begin{proof}
Consider an arbitrary ordering $i_1, i_2, \ldots, i_{|S'|}$ of $S'$. In this proof, all vectors we consider will have dimension $|S'|$ (with indices corresponding to $S'$) and hold nonnegative integer entries. For two such vectors $\vec{\psi}$ and $\vec{u}'$, we let $\vec{\psi} \leq \vec{u}'$ indicate the bound holds pointwise at each coordinate. Also let $i'_1, i'_2, \ldots, i'_{S}$ be an arbitrary ordering of $V_F$ itself. Our goal is to decide which of these to close and reroute to $S'$.

We use dynamic programming; for each $0 \leq j \leq |V_F|$ and each vector $\vec{u}'$ with $\vec{u}'_k \leq u_{i_k}, 1 \leq k \leq |S'|$, let $g(j, \vec{u}')$ be the cheapest way to close a subset $T' \subseteq \{i'_1, \ldots, i'_j\}$ and reroute $\sigma_S^{-1}(T')$ to $S'$ ensuring each $i_a \in S'$ receives at most $\vec{u}'_a$ clients this way. Of course, the final solution will be $g(|S|, \vec{u})$ where $\vec{u}_k = u_{i_k}$ for each $1 \leq k \leq |S'|$.

For brevity, for each $i'_j$ we let $\Psi_j$ be all vectors $\vec{\psi}$ such that $\sum_{k=1}^{|S'|} \vec{\psi}_k = |\sigma^{-1}_S(i'_j)$. That is, all possible ways to redistribute the $|\sigma^{-1}_S(i'_j)|$ clients currently assigned to $i'_j$ to other locations in $S'$. Note, $|\vec{\psi}| \leq n^\rho$.
For each $\vec{\psi} \in \Psi_j$, let $d_j(\vec{\psi}) = \sum_{k=1}^{|S'|} \vec{\psi}_k \cdot d(i'_j, i_k)$ be the total cost of moving the clients currently assigned to $i'_j$ to $S'$ according to $\vec{\psi}$.

\[
g(j, \vec{u}') = 
\begin{cases}
f \cdot (|S'-S|) & \text{ if } j = 0\\
\displaystyle \min_{\vec{\psi} \in \Psi_j, \vec{\psi} \leq \vec{u}'} g(j-1, \vec{u}' - \vec{\psi}) +  d_j(\vec{\psi}) & \text{ if } j \geq 1 \text{ and } i'_j \in S' \\
\displaystyle \min\left\{ g(j-1, \vec{u}'), -f + \min_{\vec{\psi} \in \Psi_j, \vec{\psi} \leq \vec{u}'} g(j-1, \vec{u}' - \vec{\psi}) +  d_j(\vec{\psi})\right\} & \text{ if } j \geq 1 \text{ and } i'_j \not\in S' \\
\end{cases}
\]

That is, if there are no further facilities in $S'$ to consider then we just have to account for the opening cost of the facilities in $S'$ that are not already open. If $i'_j \in S'$ as well then it is already open so we just try all possible reassignments of its clients.
Finally, if $i'_j \notin S'$ then we either try keeping it open or try closing it and then trying all possible client reassignments. We comment that it could be that in the two recursive cases that there is no $\vec{\psi} \in \Psi_j$ with $\vec{\psi} \leq \vec{u}'$, i.e. the leftover capacity is not able to accommodate all clients in $\sigma^{-1}(i'_j)$. In this case, an empty $\min$ takes the value $\infty$.

Since there are $n^{O(\rho)}$ subproblems and each takes $n^{O(\rho)}$ time to process, the total running time is $n^{O(\rho)}$. The recurrence only computes the minimum possible value, but the optimal $T$ for a {\bf swap}$(S',T)$ move can be recovered through an optimal trace through the dynamic programming table.
\end{proof}

\begin{lemma}
    Given $S' \subseteq V_F$ with $|S'| \leq \rho'$, we can find some $T \subseteq S$ minimizing the cost change of {\bf swap}($S',T)$ in the \ucflo algorithm in time $n^{O(\rho')}$.
\end{lemma}
\begin{proof}
Let $\vec{u}$ be the vector of capcities of facilities in $S'$. First, we guess vectors $\vec{u^1}, \vec{u^2}$ with nonnegative integer vectors, i.e. such that $\vec{u^1} + \vec{u^2} = \vec{u}$. The algorithm returns the best solution computed over all capacities.

For a particular guess, we compute a minimum-cost partial mapping from $Z_S$ to $S'$ such that exactly $\vec{u}^1_k$ clients from $Z_S$ are mapped to $i'_k$. If there is no such assignment, i.e. $\sum_k \vec{u}^1_k > z$, then we skip this guess.

Finally, we run the dynamic programming algorithm from the previous proof with minor modifications. We start with initial capacities $\vec{u}^2$ and we carry an additional parameter $z'$ that is initially $\sum_k \vec{u}^1_k$ (the number of outliers we reclaimed via $\vec{u}^1$). That is, we let $g'(j, \vec{u}', z')$ be the cheapest way to assign clients in $\sigma^{-1}_S(\{i'_1, i'_2, \ldots, i'_j\})$ to $S'$ (respecting capacities $\vec{u}'$), allowing for $z'$ of them to be outliers instead.

The recurrence in this case is nearly identical to the previous Lemma except we also guess in each step how many clients from $\sigma^{-1}(i'_j)$ to make outliers and adjust the parameters appropriately when we recurse. The straightforward details are omitted
\end{proof}

\end{document}